\documentclass[12pt,reqno,twoside]{amsart}
\usepackage{amsmath,amssymb,amsfonts,amsthm,indentfirst,setspace}
\usepackage{xcolor}\usepackage{ulem}
\usepackage{graphicx}
\usepackage{booktabs}
\usepackage{multirow}
\usepackage{hyperref}
\newtheorem{theorem}{Theorem}[section]

\newtheorem{remark}{Remark}[section]
\newcommand{\be}{\begin{equation}}
\newcommand{\ee}{\end{equation}}
\newcommand{\bea}{\begin{eqnarray}}
\newcommand{\eea}{\end{eqnarray}}
\newcommand{\eeas}{\end{eqnarray*}}
\newcommand{\beas}{\begin{eqnarray*}}

\begin{document}

\title[A CONFORMALLY FLAT GENERALIZED RICCI RECURRENT SPACETIME...]{A CONFORMALLY FLAT GENERALIZED RICCI RECURRENT SPACETIME IN $F(R)$-GRAVITY}
\author{Avik De \and Tee-How Loo \and Raja Solanki \and P.K. Sahoo}
\address{A. De\\
Department of Mathematical and Actuarial Sciences\\
Universiti Tunku Abdul Rahman\\
Jalan Sungai Long\\
43000 Cheras\\
Malaysia}
\email{de.math@gmail.com}
\address{T. H. Loo\\
Institute of Mathematical Sciences\\
University of Malaya\\
50603 Kuala Lumpur\\
Malaysia}
\email{looth@um.edu.my}
\address{Raja Solanki\\
Department of Mathematics, Birla Institute of Technology and Science-Pilani, Hyderabad Campus, Hyderabad 500078\\ 
India}
\email{rajasolanki8268@gmail.com}
\address{P.K. Sahoo\\
Department of Mathematics, Birla Institute of Technology and Science-Pilani, Hyderabad Campus, Hyderabad 500078\\ 
India}
\email{pksahoo@hyderabad.bits-pilani.ac.in}

\date{}

\thanks{A.D. and L.T.H. are supported by the grant FRGS/1/2019/STG06/UM/02/6. R.S. acknowledges University Grants Commission(UGC), Govt. of India, New Delhi, for awarding Junior Research Fellowship(NFOBC)(No.F.44-1/2018(SA-III)) for financial support.}

\begin{abstract}

In the present paper we study a conformally flat generalized Ricci recurrent perfect fluid spacetime with constant Ricci scalar as a solution of modified $F(R)$-gravity theory. We show that a Robertson-Walker spacetime is generalized Ricci Recurrent if and only if it is Ricci symmetric. The perfect fluid type matter is shown to have EoS $\omega=-1$. Some energy conditions are analyzed with couple of popular toy models of $F(R)$-gravity, like $F(R)= R+\alpha R^m $ where $\alpha, m$ are constants and $ F(R)=R+\beta RlnR $ where $\beta $ is constant. In harmony with the recent observational studies of accelerated expansion of the universe, both cases exhibit that the null, weak, and dominant energy conditions fulfill their requirements whereas the strong energy condition is violated.

\end{abstract}

\maketitle

\section{\textbf{Introduction}}\label{sec-1}

The stress-energy tensor $T_{ij}$ in the standard theory of gravity due to Einstein 
$$R_{ij}-\frac{R}{2}g_{ij}=\kappa^2 T_{ij},$$
is conserved. This requirement is accomplished if $T_{ij}$ is covariantly constant. It is shown that a spacetime solution of the standard theory of gravity with covariantly constant stress-energy tensor is Ricci symmetric, that is, $\nabla_iR_{jl} = 0$ \cite{chakiray}. In \cite{patterson}, Patterson generalized this concept by defining a Ricci recurrent manifold $R_n$, $(n> 2)$ whose Ricci tensor $R_{ij}\,(\neq 0)$ satisfies the condition $\nabla _{i}R_{jl}= A_iR_{jl}$, where $A_i\,(\neq 0)$ is a non-zero 1-form. With growing interest in the Ricci recurrent manifolds, in 1995 De et al. \cite{de} generalized the notion to introduce a generalized Ricci recurrent manifold $(GR)_n$ as a $n$-dimensional non-flat Riemannian manifold whose Ricci tensor satisfies the following:
\be \nabla _{i}R_{jl}= A_iR_{jl}+B_ig_{jl},\label{gr}\ee
where $A_i$ and $B_i$ are non-zero 1-forms. Obviously, if the one-form $B_i$ vanishes, a $(GR)_n$ reduces to a $R_n$.

Due to the ambiguity of the standard theory of gravity to explain the late time inflation of the universe, some researchers tried to modify the action principle of Einstein's field equations (EFE) to get more efficient gravitational field equations. One of the most popular of them are proposed by Hans Adolph Buchdahl in 1970 \cite{Hans}, obtained by considering a function $F(R)$ instead of the Ricci scalar $R$ in the Einstein-Hilbert action and somewhat popularised by Starobinsky by considering a quadratic form of the Ricci scalar in his formulation \cite{St}. There are several models or functional forms of $F(R)$ proposed in the literature, for reference see (\cite{O}, \cite{C}). The viability of $F(R)$-models were constrained by several observational data \cite{S.C.}.  It was shown that astrophysical structures like massive neutron stars which cannot be addressed by GR, that can be solved by the higher order curvature of $F(R)$ gravity, for references see (\cite{Aa}, \cite{Bb}, \cite{Cc}, \cite{Dd}, \cite{Ee}). The equations of motion of $F(R)$ gravity have higher degrees and provide considerable solutions that are different from general relativity. 

A Lorentzian manifold of signature $(-,+,+,+)$ is a $(GR)_4$ spacetime if the Ricci tensor satisfies (\ref{gr}). In \cite{avik}, the authors investigated a $(GR)_4$ type solution of EFE, a conformally flat $(GR)_4$ was shown to be a perfect fluid under different sets of conditions imposed on its geometry and several non-trivial examples were constructed. Recently a $(GR)_4$ spacetime with a Codazzi type Ricci tensor was shown to be a generalized Robertson Walker spacetime with Einstein fiber \cite{pinaki}, also the stress-energy tensor was proved to be of semisymmetric type in such a spacetime with constant $R$. Almost Pseudo-Ricci symmetric and Weakly Ricci Symmetric type spacetime solutions of $F(R)$-gravity were recently studied in \cite{aprs} and \cite{wrs}, respectively.   

A perfect fluid type stress-energy tensor $T_{ij}=pg_{ij}+(p+\rho)u_iu_j$ is considered, where $u_i$ denotes the four velocity vector of the fluid. In addition, we assume that the relation between the energy density $\rho$ and the pressure $p$ of the matter present in the universe is given by the equation $ p=\omega \rho $.

We know that the energy conditions represent paths to establish the positivity of the energy-momentum tensor. Also, they can be used to test the attractive nature of gravity, besides assigning the fundamental causal and the geodesic structure of space-time \cite{E}. The different models of $F(R)$ gravity give rise to the problem of constraining model parameters. By imposing the different energy conditions, we may have some constrains of $F(R)$ model parameters \cite{Gb}. The different energy conditions have been used to obtain solutions for a plenty of problems. For example the strong and weak energy conditions were used in the Hawking-Penrose singularity theorems (\cite{M.P.}, \cite{Jos}) and the null energy condition is required in order to prove the second law of black hole thermodynamics \cite{Mau}. The energy conditions were primarily formulated in GR \cite{H}, and later derived in $F(R)$ gravity by introducing effective pressure and energy density. In this paper, we examine the null, weak, strong and dominant energy conditions for $F(R)$-gravity models. In different gravity theories energy conditions are studied, for reference see (\cite{Sanjay}, \cite{Simran}). 

The present paper is organized as follows: in section 2 we find a necessary and sufficient condition for a Robertson-Walker spacetime to be generalized Ricci recurrent; followed by a study of conformally flat $(GR)_4$ with constant $R$ which satisfies $F(R)$-gravity equations. Several energy conditions in such a spacetime are discussed next, followed by couple of popular $F(R)$-gravity models, analysed in conformally flat $(GR)_4$. In section 5 we consider two models, one is polynomial and the other logarithmic and investigate different energy conditions and find the constraints on model parameters to satisfy requirements of energy conditions. Finally in section 6 we end up with the discussion.

\section{\textbf{Robertson-Walker spacetime as a $(GR)_4$}}\label{sec-2}
In this section we show that a spatially flat RW spacetime is generalized Ricci recurrent under certain conditions. 

In a spatially flat RW spacetime we have
\begin{align*}
ds^2=-dt^2+a^2(t)\left(dr^2+r^2d\theta^2+r^2\sin^2\theta d\phi^2\right),
\end{align*}

and
\begin{align}\label{eqn:7a}
R_{jl}=(P-Q)u_ju_l+Pg_{jl}=-Qu_ju_l+Ph_{jl}
\end{align}
where 
\begin{align}\label{eqn:2a}
P=\frac{a\ddot a+2\dot a^2}{a^2}, \quad Q=3\frac{\ddot a}{a} 
\end{align}
and  $u^i=(\partial_t)^i$ is the four-velocity of the fluid with $u^ju_j=-1$ and  
\begin{align}\label{eqn:3}
\nabla_ju_l=&\frac{\dot a}ah_{jl}.   
\end{align}
It is clear from (\ref{eqn:2a}) that 
\begin{align}\label{eqn:7b}
\nabla_iP=&-u_i\dot P, \quad \nabla_iQ=-u_i\dot Q.
\end{align} 
Taking covariant derivative on (\ref{eqn:7a}), with the help of (\ref{eqn:3})--(\ref{eqn:7b}) we obtain
\begin{align}\label{eqn:9a}
\nabla_iR_{jl} 
=&(\nabla_iP-\nabla_iQ)u_ju_l+(P-Q)\{\nabla_iu_ju_l+\nabla_iu_lu_j\}+\nabla_iPg_{jl} \notag \\
=&\dot Qu_iu_ju_l+(P-Q)\frac{\dot a}a\{h_{ij}u_l+h_{il}u_j\}-\dot Pu_ih_{jl}.	  
\end{align}

If it is also $(GR)_4$, by (\ref{gr}) and (\ref{eqn:7a}), the above equation gives
\begin{align}\label{eqn:9b}
\nabla_iR_{jl}
=& -(QA_i+B_i)u_ju_l +(PA_i+B_i)h_{jl}.
\end{align}
By comparing (\ref{eqn:9a})--(\ref{eqn:9b}), we have
\begin{align}\label{eqn:10}	
\{\dot Qu_i +QA_i+B_i)\}u_ju_l+(P-Q)\frac{\dot a}a\{u_lh_{ij}+u_jh_{il}\} \notag\\
=\{\dot Pu_i+PA_i+B_i\}h_{jl} 
\end{align}
Let $h^{jl}=g^{jl}+u^ju^l$. Then
\[
h_{ij}h^{jl}=h_i^l=\delta_i^l+u_iu^l.
\]
Transvecting (\ref{eqn:10}) with $h^{jl}$, we have
\begin{align}\label{eqn:12}\dot Pu_i+PA_i+B_i=0.\end{align}
Transvecting (\ref{eqn:10}) with $u^jh^{il}$, we have
\begin{align}\label{eqn:12b}
P-Q=0
\end{align}
or equivalently, $a\ddot a=\dot a^2$. Solving this equation gives 
\[
a(t)=Ae^{\epsilon t}
\]
where $A>0$ and $\epsilon$ are constants. 
By using this and (\ref{eqn:2a}), we have $\dot P=0$ and so 
(\ref{eqn:12}) gives
\begin{align}\label{eqn:12d}
PA_i+B_i=0.
\end{align}
By substituting all these in (\ref{eqn:9b}), we have $\nabla_iR_{jl}=0$, that is the spacetime is Ricci symmetric. Combining, we conclude:
\begin{theorem}
A spatially flat RW spacetime is a $(GR)_4$ spacetime if and only if 
 it is Ricci symmetric and the scale factor $a(t)$  is given by
\[
a(t)=Ae^{\epsilon t}
\]
where $A(>0)$ and $\epsilon$ are constants and $3\epsilon^2A_i+B_i=0$.
\end{theorem}

\section{\textbf{Conformally flat $(GR)_4$ satisfying $F(R)$-gravity}}\label{sec-3}

In this section we study a conformally flat $(GR)_4$ with a constant $R$. Contraction of $j,l$ in (\ref{gr}) produces
\be \nabla_iR=A_iR+4B_i.\label{dr}\ee
Since $R$ is constant, from (\ref{dr}) we obtain, 
\be B_i=-\frac{R}{4}A_i.\label{AB}\ee
On the other hand, since the spacetime is conformally flat, we have the well-known relation
\beas\nabla_iR{jk}-\nabla_kR_{ij}=\frac{1}{6}[g_{jk}\nabla_iR-g_{ij}\nabla_kR],\eeas
which, for a constant $R$, reduces to 
\be \nabla_iR_{jk}=\nabla_kR_{ji}.\label{14}\ee
Using (\ref{gr}) and (\ref{AB}), from (\ref{14}) we get
\be 
A_iR_{jk}-A_kR_{ij}=-\frac{R}{4}(A_kg_{ij}-A_ig_{jk}).\ee
Contracting $j,k$ in the above equation we finally get 
\be A^jR_{ij}=\frac{R}{4}A_i.\label{barb}\ee

From an action term \[S=\frac{1}{\kappa^2}\int F(R) \sqrt{-g}d^4x +\int L_m\sqrt{-g}d^4x,\]
we obtain the field equations
\be 
F_R(R)R_{ij}-\frac{1}{2}F(R)g_{ij}+(g_{ij}\Box-\nabla_i\nabla_j)F_R(R)=\kappa^2T_{ij},\label{FR}	
\ee
where $\Box=\nabla^k\nabla_k$, $L_m$ is the matter Lagrangian, and 
\[T_{ij}=-\frac{2}{\sqrt{-g}}\frac{\delta(\sqrt{-g}L_m)}{\delta g^{ij}},\] . 

For a constant Ricci scalar, (\ref{FR}) reduces to:
\be R_{ij}-\frac{R}{2}g_{ij}=\frac{\kappa^2}{F_R(R)}T^{\text{eff}}_{ij},\label{fr}\ee
where 
$$T^{\text{eff}}_{ij}=T_{ij}+\frac{F(R)-RF_R(R)}{2\kappa^2}g_{ij}.$$ 
If a perfect fluid $(GR)_4$ spacetime satisfies (\ref{fr}) with the velocity vector identical to $A^i$ (assuming a time-like unit vector), we have 
\be R_{ij}=\frac{\kappa^2(p+\rho)}{F_R(R)}A_iA_j+\frac{2\kappa^2p+F(R)}{2F_R(R)}g_{ij}.\label{a}\ee

This readily gives us
\be R_{ij}A^j=\left[\frac{F(R)-2\rho \kappa^2}{2F_R(R)} \right]A_i.\label{a1}\ee
But from (\ref{barb}) we already know that $R_{ij}A^j=\frac{R}{4}A_i$. Hence, we conclude that 
\be R=\frac{2F(R)-4\rho \kappa^2}{F_R(R)}.\ee
Therefore, \be \rho=\frac{2F(R)-RF_R(R)}{4k^2},\ee
which from the trace equation of (\ref{a}) also gives us 
\be p=-\frac{2F(R)-RF_R(R)}{4k^2}.\ee
This leads to our first result of this section:
\begin{theorem}\label{pfthm}
In a $(GR)_4$ spacetime solution of $F(R)$-gravity with constant $R$, if the four-velocity vector is identical with $A^i$; then $p$ and $\rho$ are constants and satisfy the relations $p=-\frac{2F(R)-RF_R(R)}{4k^2}$ and $\rho=\frac{2F(R)-RF_R(R)}{4k^2}$. 
\end{theorem}
\begin{remark}
The equation of state $w=-1$ is fulfilled in this case, consistent with the presently well-established $\Lambda$CDM theory. This theory is supported both theoretically and also by the plethora of observational data in recent years.
\end{remark}
\begin{theorem}
A vacuum $(GR)_4$ solution with constant $R$ in $F(R)$-gravity is only possible when $F(R)$ is a constant multiple of $R^2$. 
\end{theorem}
\begin{proof}
For the vacuum case, $T_{ij}=0$ in (\ref{fr}), the trace equation gives $RF_R(R)=2F(R)$ which on integration gives $F(R)=\lambda R^2$ for integrating constant $\lambda$.
\end{proof}
\section{\textbf{Energy conditions in a $(GR)_4$}}\label{sec-4}

In the standard theory of gravity and the modified gravity theories, energy conditions acts as a filtration system to constraint the (effective) stress-energy tensor. For the current study of modified $F(R)$- theories of gravity, we first deduce the effective pressure $p^{\text{eff}}=-\frac{RF_R(R)}{4k^2}$ and the effective energy density $\rho^{\text{eff}}=\frac{RF_R(R)}{4k^2}$ from (\ref{fr}) to investigate some energy conditions as follows:
\begin{itemize}
\item Null energy condition \textbf{(NEC)}: $\rho^{\text{eff}}+p^{\text{eff}}\geq 0$. Hence the NEC is always satisfied in the present setting. 

\item Weak energy condition \textbf{(WEC)}: $\rho^{\text{eff}}\geq 0$ and $\rho^{\text{eff}}+p^{\text{eff}}\geq 0$. Hence, WEC is satisfied in the present setting if $RF_R(R)\geq 0$. Since $F_R(R)> 0$, it implies that $R\geq 0$.

\item Dominant energy condition \textbf{(DEC)}: $\rho^{\text{eff}}\pm p^{\text{eff}}\geq 0$. In the present setting, DEC is only satisfied if $2 RF_R(R)\geq 0,$ which implies a non-negative $R$ as in the previous case of WEC, since $F_R(R)>0$.  

\item 
Strong energy condition \textbf{(SEC)}: $\rho^{\text{eff}}+3p^{\text{eff}}\geq 0$. In the present setting this is only satisfied if $RF_R(R)\leq 0$. Since $F_R(R)>0$, it implies $R\leq 0$.  
\end{itemize}
However, recent observational data of accelerating inflation of the universe strongly disfavour the SEC on cosmological scales. Therefore, we propose a positive Ricci scalar $R$ in the present study.


\section{\textbf{Analysis of $F(R)$-gravity models in $(GR)_4$}}\label{sec-5}


In this section we analyse two different models of $F(R)$-gravity theories to analyse our result.
\\

\textbf{Case:I}  
$F(R) = R+\alpha R^m $ , $\alpha $ and $ m $ are constant. Some of the polynomial models are studied in \cite{Santos}. In this case, the effective stress-energy tensor reduces to,
\\
\begin{equation}
T^{\text{eff}}_{ij} = T_{ij} + \frac{1}{2 \kappa^2} (1-m) \alpha R^m g_{ij},
\end{equation}
\\
which gives
\\
\begin{equation}
  \rho^{\text{eff}} = \frac{R+ m \alpha R^m}{4 \kappa^2} 
\end{equation}
and 
\begin{equation}
 p^{\text{eff}} = - \frac{R+ m \alpha R^m}{4 \kappa^2}. 
\end{equation}
\\
The equation of state parameter (EoS) in this case reads as $\omega = -1 $ for any $\alpha$ and $m$, that is, the universe is dominated by cosmological constant and this model is consistent with the presently well established $\Lambda CDM$ theory. As WEC is a combination of NEC and positive density and NEC is always zero in the present setting, so we observe behaviour of SEC, DEC and density parameter.

\begin{figure}[h]
{\includegraphics[scale=0.5]{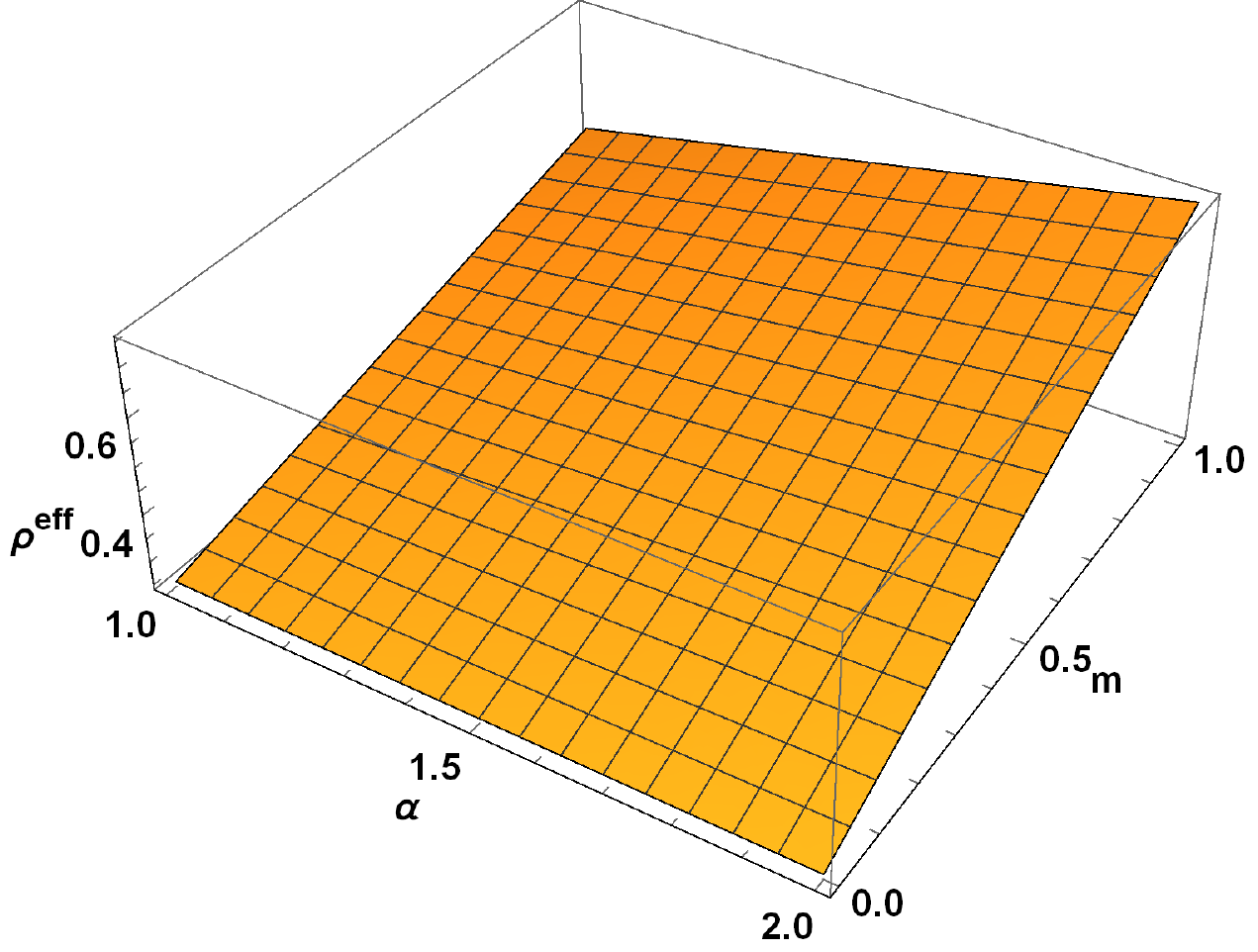}}
{\includegraphics[scale=0.5]{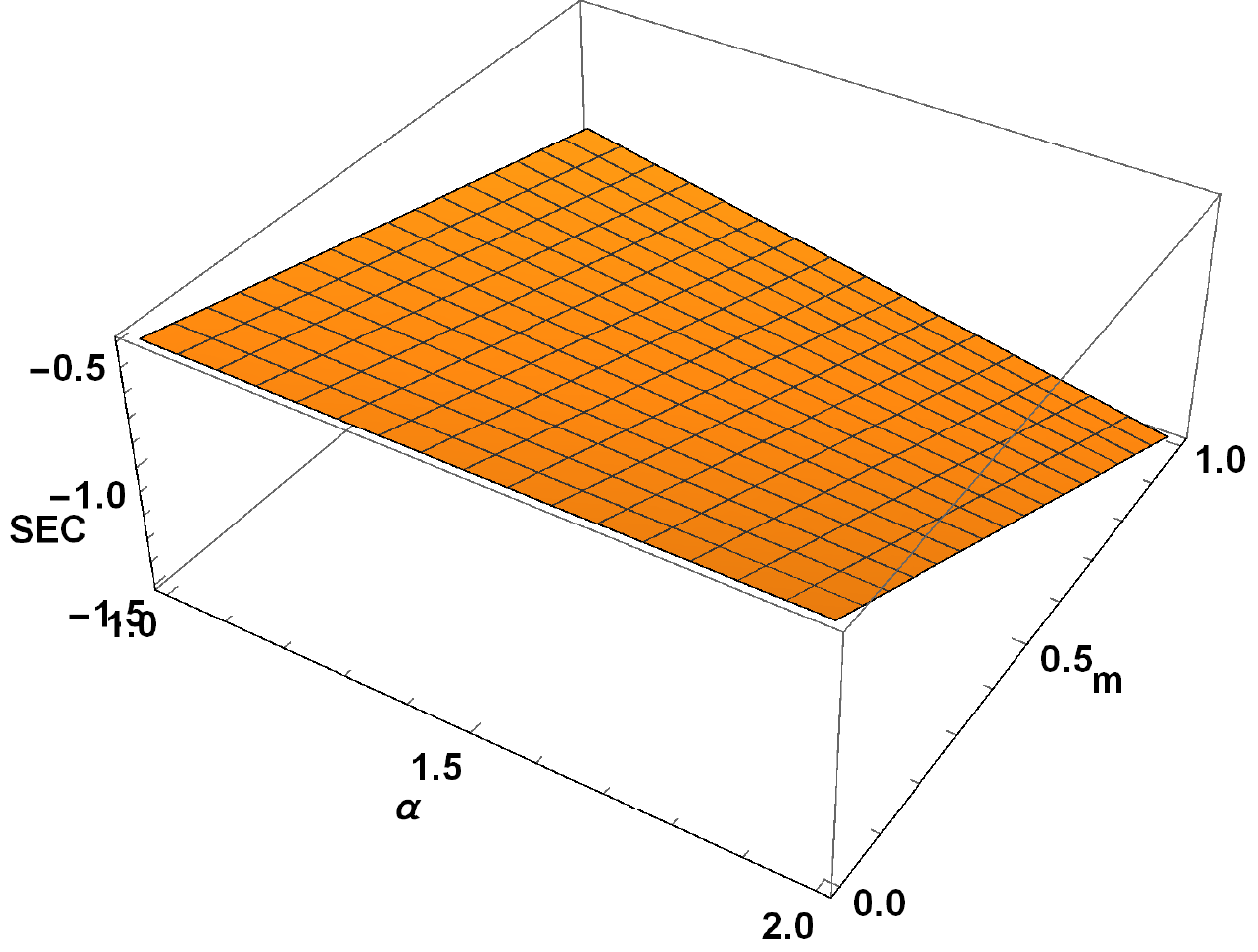}}
{\includegraphics[scale=0.5]{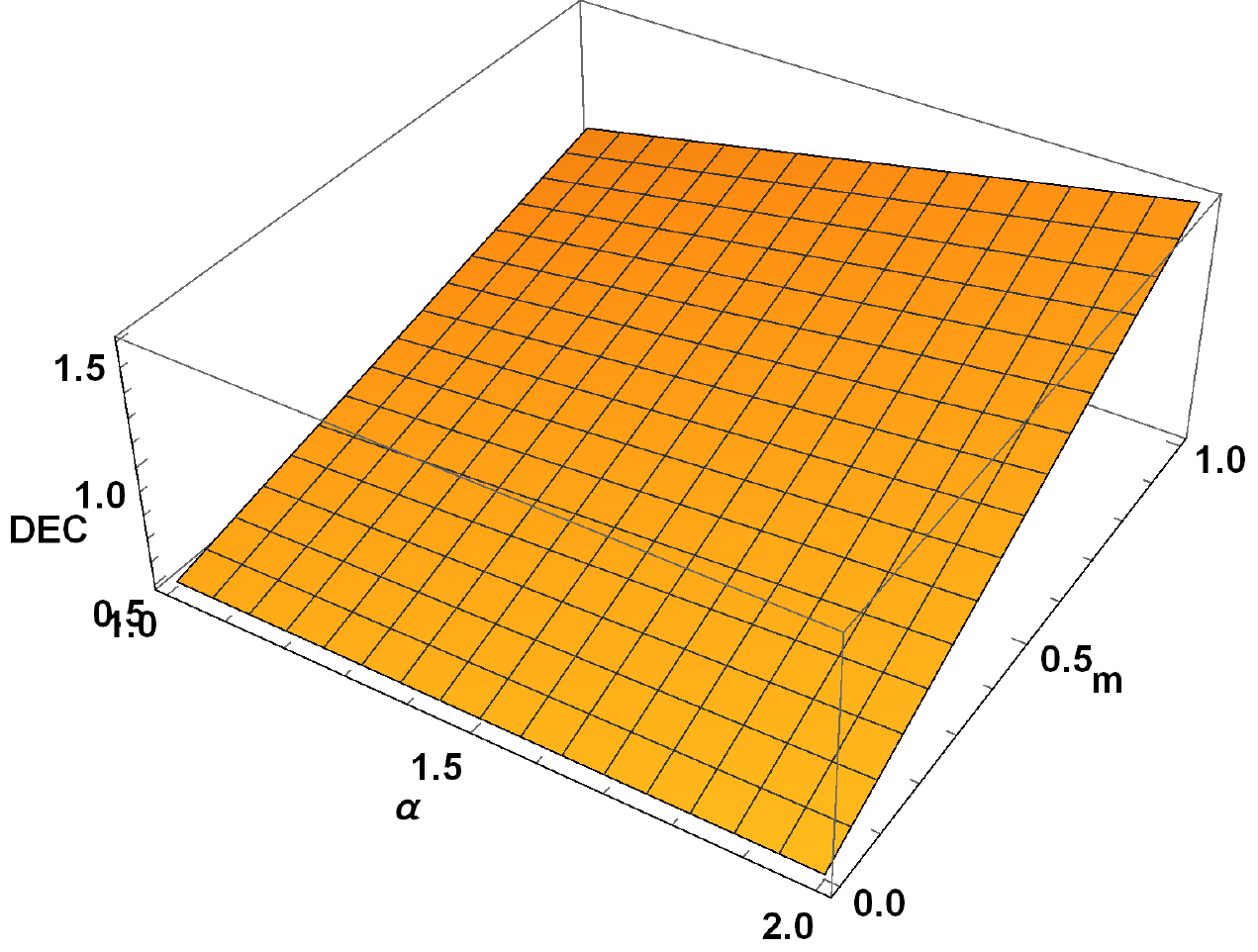}}
\caption{Energy conditions for $F(R)=R+ \alpha R^m $ with $1\leq \alpha \leq 2$,   $ 0 \leq m \leq 1$ and $R=1$.}\label{fig-1}
\end{figure}

In these figures [\ref{fig-1}]  we plot the density parameter, SEC and DEC with respect to $\alpha$ and $ m $ taking $R=1 $, $ 1 \leq \alpha \leq 2 $ and $ 0 \leq m \leq 1$.\\

\newpage
 We observe the behavior of all energy conditions and density parameter for $R>0$ , the summary of the results obtained is given below.
\begin{table}[h!]
\caption{Constraints on model parameters $\: m \:$ and $ \: \alpha \: $ to satisfy different energy conditions and violating SEC. }
\begin{center}
\begin{tabular}{c|c|c}
\toprule
\textbf{S. no.} & \textbf{Terms} & \textbf{Results} \\
\midrule
\multirow{2}{0.5em}{1} &  & for $m\in$ $(0,\infty), \alpha \in (0,\infty) $ \\
&   $\rho^{eff} >0 $ &  and for  $ m\in (- \infty,0), \alpha \in  (-\infty,0)$ \\ 
& & and for $ m=0$, $\alpha \epsilon (-\infty,\infty) $\\
\hline
\multirow{2}{0.5em}{2} &  & for $m\in$ $(0,\infty), \alpha \in (0,\infty) $ \\
&  WEC, DEC$ >0$ and SEC$<0$  &  and for  $ m\in (- \infty,0), \alpha \in  (-\infty,0)$  \\ 
&  & and for $ m=0$, $\alpha \epsilon (-\infty,\infty) $\\
\hline
3 &  { NEC = 0} & for $m\in$ $(-\infty,\infty), \alpha \in (-\infty,\infty) $  \\
\bottomrule 
\end{tabular}
\end{center}  
\label{table-1}
\end{table}

Clearly WEC and DEC satisfy the condition of positivity while SEC violates it, which imply the accelerated expansion of the universe.
\\

\textbf{Case:II}
$F(R)= R+\beta R ln(R) $, $ \beta $ is constant. Some of the logarithmic models are studied in \cite{Micol}. In this case, the effective energy momentum reduces to,
\\
\begin{equation}
T^{\text{eff}}_{ij} = T_{ij} - \frac{\beta R}{2 \kappa^2} g_{ij} 
\end{equation}
\\
which gives,
\\
\begin{equation}
\rho^{\text{eff}} = \frac{R(1+\beta+\beta ln(R))}{4 \kappa^2}
\end{equation}
and 
\begin{equation}
 p^{\text{eff}} = - \frac{R(1+\beta+\beta ln(R)) }{4 \kappa^2}. 
\end{equation}
\\
The equation of state parameter (EoS) reads as $\omega = -1 $ for any $\beta$ and $ R>0 $ and as previous section, NEC is always zero in the present setting, so we observe behavior of SEC, DEC and density parameter.
The summary of the results obtained from analysis of all energy conditions and density parameter is given below.
\begin{table}[h!]
\caption{Constraints on model parameter $\: \beta \:$ and $\:R\:$ to satisfy different energy conditions and violating SEC. }
\begin{center}
\begin{tabular}{c|c|c}
\toprule
\textbf{S. no.} & \textbf{Terms} & \textbf{Results} \\
\midrule
\multirow{2}{0.5em}{1} &  & for $R\in$ $(0,\frac{1}{e}), \beta \in (-\infty,- \frac{1}{1+lnR}) $ \\ 
&  $\rho^{eff} >0 $  & \\
& & and for  $ R \in (\frac{1}{e},\infty), \beta \in  (-\frac{1}{1+lnR},\infty)$ \\ 
& & \\
& & and for $ R=\frac{1}{e}$, $\beta \in (-\infty,\infty) $\\ 
& & \\
\hline
\multirow{2}{0.5em}{2} &  & for $R\in$ $(0,\frac{1}{e}), \beta \in (-\infty,- \frac{1}{1+lnR}) $ \\
&  WEC, DEC$ >0$ and SEC$<0$ & \\
& & and for  $ R \in (\frac{1}{e},\infty), \beta \in  (-\frac{1}{1+lnR},\infty)$ \\ 
& & \\
& & and for $ R=\frac{1}{e}$, $\beta \in (-\infty,\infty) $\\
& & \\
\hline
3 &  { NEC = 0} & for any $R>0, \beta \in (-\infty,\infty) $  \\
\bottomrule 
\end{tabular}
\end{center}
\label{table-2}
\end{table}
\\ 

\newpage
\begin{figure}
{\includegraphics[scale=0.5]{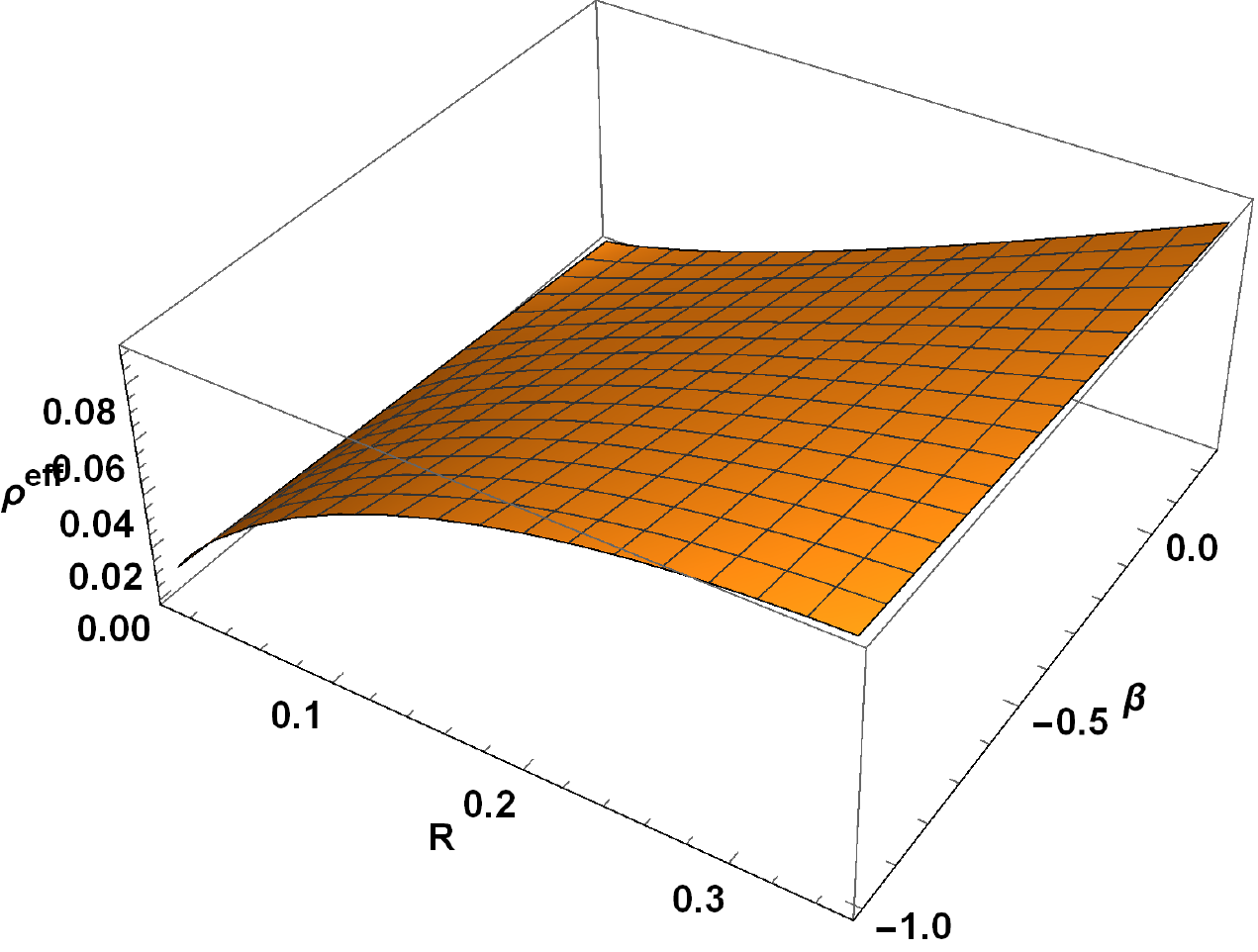}}
{\includegraphics[scale=0.5]{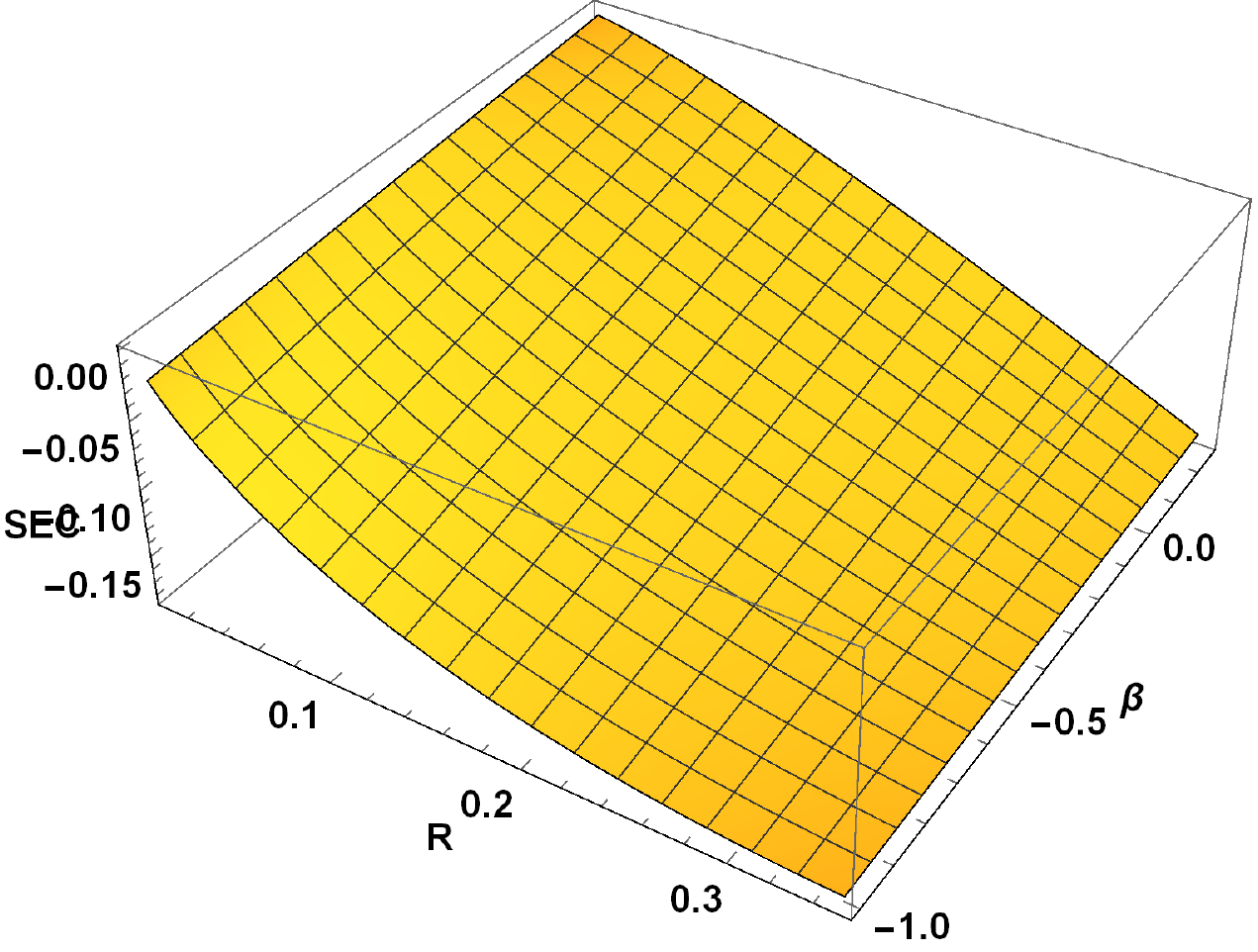}}
{\includegraphics[scale=0.5]{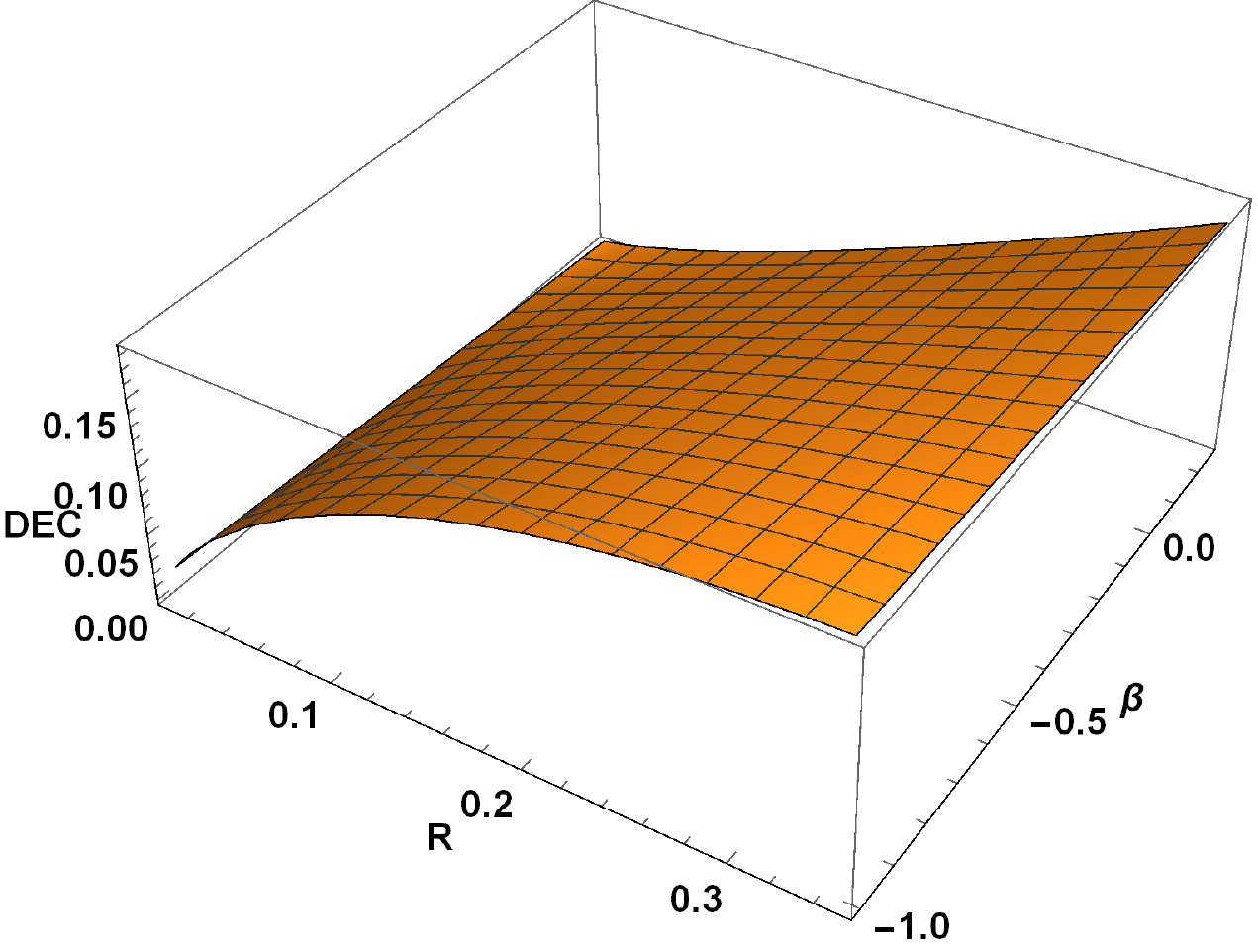}}
\caption{Energy conditions for $F(R)=R+ \beta ln (R) $ with $-1\leq \beta \leq 0.5$ and   $ 0.01 \leq R \leq 0.35$.}\label{fig-2}
\end{figure}

The above figures [\ref{fig-2}] are plotted by taking a range of $ 0.01\leq R \leq 0.35 $ and $ -1 \leq \beta \leq 0.25 $. From the figure \ref{fig-2} it is clear that WEC and DEC satisfies the  condition of positivity while SEC violates it and NEC is always zero. The violation of SEC complies with the accelerated expansion of the universe which is compatible with recent observational studies.
\\

\section{\textbf{Discussion}}\label{sec-6}


The different energy conditions which can be derived from the well known Raychaudhuri equation plays an important role to define self-consistency of modified theories of gravity. In this paper we examine the null, weak, dominant and strong energy conditions for the $F(R)$-gravity under conformally flat generalized Ricci recurrent perfect fluid spacetime with constant Ricci scalar. In two different $F(R)$ models, $F(R)= R+\alpha R^m $ where $\alpha, m$ are constants and $ F(R)=R+\beta RlnR $ where $\beta $ is constant, we investigate various energy conditions. The results from section \ref{sec-4} are used to constrain the model parameters in these $F(R)$ models. $R>0$ and constraints on the other parameters are given in the tables of corresponding models. We illustrate that density parameter and DEC fulfill their requirements of positivity when $ 1 \leq \alpha \leq 2 $, $ 0 \leq m \leq 1 $, $R=1$ and $ -1 \leq \beta \leq 0.25 $ , $ 0.01 \leq R \leq 0.35 $. The strong energy condition violates its condition of non-negativity and shows the negative behaviour in both the cases with the given constraints on model parameters and which favours the accelerated expansion scenario of the universe.

\end{document}